\documentclass[12pt]{article}
\usepackage{eurosym}
\usepackage{amssymb}
\usepackage{amsfonts}
\usepackage[onehalfspacing]{setspace}

\newtheorem{theorem}{Theorem}

\newtheorem{corollary}[theorem]{Corollary}

\newtheorem{definition}{Definition}

\newtheorem{lemma}{Lemma}

\newtheorem{proposition}{Proposition}
\newtheorem{remark}[theorem]{Remark}

\newenvironment{proof}[1][Proof]{\noindent\textbf{#1.} }{\ \rule{0.5em}{0.5em}}
\input{tcilatex}
\begin{document}

\title{Competitive Markets with Imperfectly Discerning Consumers\thanks{%
Financial support from UKRI Frontier Research Grant no. EP/Y033361/1 is
gratefully acknowledged. We thank Meg Meyer, Marco Scarsini, Benny Moldovanu
and Eylon Solan for helpful comments.}}
\author{Yair Antler and Ran Spiegler\thanks{%
Antler: Tel Aviv University. Spiegler: Tel Aviv University and University
College London}}
\maketitle

\begin{abstract}
We develop a market model in which products generate state-dependent
potential hidden charges. Firms differ in their ability to realize this
potential. Unlike firms, consumers do not observe the state. They try to
infer hidden charges from market prices, using idiosyncratic subjective
models. We show that an interior\ competitive equilibrium is uniquely given
by what is formally a Bellman\ equation. We leverage this representation to
characterize equilibrium headline prices, add-on charges and welfare. Market
responses to shocks display patterns that are impossible under rational
expectations. For example, equilibrium prices can be fully revealing and yet
vary with consumers' private information.\bigskip \bigskip \bigskip \bigskip
\pagebreak 
\end{abstract}

\section{Introduction}

One of the deepest ideas in the history of economic thought is that
competitive markets aggregate private information through the price
mechanism (Hayek (1945), Radner (1979)). According to this idea,
competitive-equilibrium prices signal unobserved payoff-relevant features.
Under mild assumptions, market participants can perfectly invert the
equilibrium price signal and effectively make informed choices, as if all
payoff-relevant information were public.

This logic relies heavily on the assumption that market participants can
make flawless inferences from equilibrium prices. However, in reality, some
participants (particularly consumers) possess a limited grasp of the
systematic relation between prices and latent variables. For example,
consumers may have a broad sense that price and quality are correlated, or
that a low headline price implies hidden costs (as captured by mottos like
\textquotedblleft there is no such thing as a free lunch\textquotedblright\
or \textquotedblleft If you're not paying for the product, you are the
product\textquotedblright ), yet lack a precise understanding of such
relations.

In this paper, we develop a model of a competitive market in which consumers
draw imperfect inferences from equilibrium prices. In our model, consumers
who purchase a product pay a headline price, as well as a hidden charge that
is an endogenous response by firms to an unobserved state of nature. In
equilibrium, there are no \textquotedblleft free lunches\textquotedblright :
A low headline price tends to signal a high hidden charge. However,
consumers differ in their ability to infer latent features from the headline
price: Fully rational consumers make perfect inferences, while others
perform imperfect inferences based on idiosyncratic, imprecise subjective
models. In other words, consumers are \textit{diversely discerning}.

Our model enables us to reexamine classic economic questions: How do
competitive forces determine the types of firms that survive in the market?
Could a rational market participant deduce latent information from
equilibrium prices? The model also raises new questions directly related to
consumers' partial and diverse sophistication: Does the presence of
imperfectly discerning consumers have a systematic effect on prices,
allocations, and welfare? How does it shape the distribution of equilibrium
payments between salient and latent components? Does it affect the way
market outcomes respond to external shocks?

The supply side in our model consists of a continuum of firms that observe
an aggregate state defined by a collection of exogenous variables, before
deciding whether to offer a basic product at a cost. Each state defines a
distinct potential \textquotedblleft add-on charge\textquotedblright . Firms
differ in their ability to realize this potential. Specifically, when a
consumer buys from a type-$\pi $ firm in some state, he pays the firm a
headline price as well as an add-on charge which is a fraction $\pi $ of the
potential add-on in that state. We assume that $\pi $ is uniformly
distributed, which conveniently generates a linear supply function.

Heterogeneity in $\pi $ can reflect differences in firms' ingenuity in
devising hidden fees. This view of exploitative hidden charges is in the
spirit of Heidhues et al. (2016), who regard them as fruits of firms'
initiative against market constraints. Other sources of heterogeneity among
firms is different exposure to regulatory restrictions, or different ex-post
opportunities for consumers to substitute away from the latent charges.
Likewise, there are several interpretations for the multiple state variables
that determine the potential add-on charge. First, there may be several
channels for hidden charges (e.g., bank fees for various financial
services); a state variable can indicate the feasibility of a particular
channel. Second, state variables can represent regulations that constrain
hidden fees. Finally, a state variable can indicate whether buying the basic
product enables the firm to extract specific private information from the
consumer and then use it to its own advantage at the consumer's expense.

The demand side in our model consists of a continuum of consumers. Each
consumer knows his idiosyncratic bare willingness to pay for the product.
However, he is uninformed about the state or the type of firm he will buy
from, and therefore aims to infer the expected add-on charge from the market
(headline) price. Consumers are classified into \textquotedblleft cognitive
types\textquotedblright\ (there is a large measure of each type). A consumer
of type $M$ considers only a subset $M$ of the state variables, ignoring the
rest (because he is unaware of them or deems them irrelevant). The consumer
infers the state variables in $M$ from the market price and forms an
estimate of the expected add-on charge based on this inference.

When $M$ omits state variables, this is a model of \textquotedblleft coarse\
beliefs\textquotedblright\ in the spirit of Eyster and Rabin (2005), Jehiel
(2005), or Eyster and Piccione (2013). More specifically, its account of how
consumers infer add-ons from prices is reminiscent of Mailath and
Samuelson's (2020) \textquotedblleft model-based inference\textquotedblright
. In Section 5, we present a more general formalism of belief formation by
Spiegler (2016), which is based on the notion of \textit{subjective causal
models}. This formalism subsumes coarse beliefs\ as a special case; our
choice to focus on the latter in our basic model serves expositional
simplicity.

\textit{Competitive equilibrium} is an assignment of prices to states, such
that each consumer optimizes with respect to his subjective belief given the
headline price; each firm offers the product if and only if this is
profitable given the headline price, the state, and the firm's type $\pi $;
and supply equals demand in every state. An equilibrium is \textit{interior}
if there are both active and inactive firms in each state. As long as there
is variation in consumers' bare willingness to pay for the product (i.e.,
demand is downward sloping), interior equilibria are \textit{fully revealing}%
, such that rational consumers can infer the expected add-on from the
equilibrium price. By comparison, boundedly rational consumers can only
partially decipher the signal that equilibrium prices provide, and therefore
form wrong add-on estimates.

Our subsequent analysis of interior equilibrium focuses on the limit case in
which the variation in consumers' bare willingness to pay is negligible. In
this limit, since there are many consumers of each type, the equilibrium
price in each state is driven by the cognitive type with the lowest estimate
of the expected add-on (which itself is inferred from the equilibrium
price). An interior equilibrium exists for a range of values of our
primitives. Moreover, it is \textit{uniquely} characterized by what is
formally a \textit{Bellman equation} --- as if \textquotedblleft the
market\textquotedblright\ tries to minimize a discounted average of add-on
charges across states.

We use the \textquotedblleft Bellman\textquotedblright\ characterization to
probe the structure of interior equilibrium. First, we analyze the effects
of expanding the set of cognitive types --- e.g., when \textquotedblleft
coarse\textquotedblright\ consumers are introduced into a population of
rational consumers. The expected equilibrium add-on weakly \textit{decreases}
in every state, while the headline price weakly \textit{rises}. Thus, making
the consumer population cognitively more diverse shifts equilibrium payments
from latent to salient components. We use this finding to show that compared
with rational-expectations equilibrium, the expected headline price in
interior equilibrium of our model is weakly higher, and the expected add-on
charge is weakly lower. The lowest possible expected add-on is obtained when
all firm types are active in every state. We show that this lower bound can
be approximated in equilibrium under a suitable selection of primitives.

The mechanism behind these results is that expanding the set of cognitive
types weakly raises the maximal net willingness to pay in each state. This
increase in demand leads to higher headline prices, which in turn alter the
composition of firms in the market. Lower-$\pi $ firms (which are less adept
at devising exploitative hidden charges) enter the market, i.e., the volume
of trade increases and the average add-on decreases. Note that the former
effect enhances social welfare, as trade is socially beneficial in our
model, but it comes at the expense of boundedly rational consumers who
incorrectly anticipate an even lower add-on charge and consequently suffer a
welfare loss.

Another implication of the \textquotedblleft Bellman\textquotedblright\
characterization concerns the $range$ of price components in equilibrium.
The range of add-on charges is narrower relative to the
rational-expectations benchmark; and when rational consumers are present,
the range of equilibrium headline prices is wider. This is another sense in
which consumers' bounded rationality shifts the \textquotedblleft
action\textquotedblright\ in the market from hidden to visible price
components.

Our modeling framework also allows us to analyze markets in which hidden
product features are not exploitative but \textit{mutually beneficial}. For
example, consider a follow-up service contract, or an aspect of product
quality that also boosts the firm's reputation. To capture these scenarios,
we modify the basic model by assuming that each state generates a distinct
latent surplus that $both$ consumer and firm enjoy. The unique interior
equilibrium is characterized by a quasi-Bellman equation like the one we
derive for the basic model, except that the discount factor is \textit{%
negative}. This sign difference carries distinct equilibrium implications:
Expanding the set of cognitive types need not have a uniform effect on
latent payoffs across states; and it \textit{lowers} social welfare when
rational consumers are present in the market.

As mentioned above, our results readily extend to the more general model in
which consumers form beliefs according to subjective causal models
(represented by so-called \textit{perfect} directed acyclic graphs). This
more general formalism can capture wider varieties of imperfectly discerning
consumers (e.g., understanding that demand for add-ons drives hidden
charges, while failing to realize that it also influences the product's
headline price). It also generates novel market responses to external
shocks. First, supply and demand responses to shocks can be virtually
independent even when the direct payoff implications of these shocks are
perfectly correlated across the two sides of the market. Second, when
consumers receive private signals about states, equilibrium\ prices can
reflect these signals in addition to the payoff-relevant state, even though
prices perfectly reveal the latter. Thus, the presence of imperfectly
discerning consumers can lead to excessive price fluctuations, in the sense
that prices respond to factors beyond economic fundamentals.

\section{The Model}

Consider a market for a product with salient and latent components. Let $%
\phi $ denote the (headline) price at which the product is traded. We now
describe supply and demand in this market, and then define competitive
equilibrium.\bigskip

\noindent \textit{Supply}

\noindent There is a measure one of firms. When a firm enters the market, it
offers a product at a cost $c$. Let $\Theta =\Theta _{1}\times \cdots \times
\Theta _{n}$ be a finite set of exogenous \textit{states}. Let $\mu \in
\Delta (\Theta )$ be a full-support distribution over states. Let $S:\Theta
\rightarrow 
%TCIMACRO{\U{211d} }%
%BeginExpansion
\mathbb{R}
%EndExpansion
_{++}$ be a one-to-one function. Denote $S^{\max }=\max_{\theta }S(\theta )$%
, $S^{\min }=\min_{\theta }S(\theta )$, and $\bar{S}=\sum_{\theta }\mu
(\theta )S(\theta )$. The quantity $S(\theta )$ represents the maximal
potential hidden charge in state $\theta $. A firm's \textit{type} is $\pi
\sim U[0,1]$, representing the firm's ability to realize the exploitative
potential. When a consumer purchases a product from a firm of type $\pi $ in
state $\theta $, a subsequent transfer of $\pi S(\theta )$ from the consumer
to the firm (in addition to the price $\phi $) is realized. We often refer
to this transfer as an \textit{add-on}. It is hidden, in the sense that
consumers do not observe it when purchasing the product.

A firm of type $\pi $ enters the market in state $\theta $ given the price $%
\phi $ if and only if it earns a non-negative profit, i.e.,%
\begin{equation}
\phi -c+\pi S(\theta )\geq 0  \label{active condition}
\end{equation}%
Let $\pi ^{\star }(\theta ,\phi )$ be the value of $\pi $ that satisfies (%
\ref{active condition}) bindingly. Total supply under $(\theta ,\phi )$ is
the measure of active firms, which is equal to $1-\pi ^{\star }(\theta ,\phi
)$ (as long as $\pi ^{\star }(\theta ,\phi )\in \lbrack 0,1]$). Thanks to
the assumption that $\pi $ is uniformly distributed, we obtain a linear
supply function in each state. The add-on value among active firms given $%
(\theta ,\phi )$ is thus a random variable, denoted $q$ and distributed as
follows:%
\begin{equation}
q\mid (\theta ,\phi )\sim U\left[ \pi ^{\star }(\theta ,\phi )S(\theta
),S(\theta )\right]   \label{quality supply}
\end{equation}%
\bigskip The expected add-on given $(\theta ,\phi )$ is%
\begin{equation}
\bar{q}(\theta )=\frac{1+\pi ^{\star }(\theta ,\phi )}{2}S(\theta )
\label{quality}
\end{equation}%
\bigskip 

\noindent \textit{Demand}

\noindent There is a large population of consumers. Let $v$ be the
consumer's bare valuation of the product, and assume it is distributed
continuously over $[v^{\ast }-\varepsilon ,v^{\ast }+\varepsilon ]$. When
the consumer buys the product in state $\theta $ at a headline price $\phi $
from a type-$\pi $ firm, his net payoff is $v-\phi -\pi S(\theta )$. Each
consumer is informed of his $v$, yet he is uninformed of $\theta $ and the
type $\pi $ of the firm he interacts with when buying a product. He tries to
infer the expected latent add-on from the market price. We will usually
assume that $\varepsilon $ is small, such that consumer preferences are
nearly homogenous. The social surplus generated by the product in the $%
\varepsilon \rightarrow 0$ limit is $\Delta =v^{\ast }-c>0$. 

Let $\mathcal{M}$ be a finite set of \textquotedblleft cognitive
types\textquotedblright . The measure of consumers of each type is greater
than one. Every $M\in \mathcal{M}$ is a distinct subset of the set $%
\{1,...,n\}$ of exogenous variables. The interpretation is that a type-$M$
consumer is unaware of variables outside $M$, or deems them irrelevant.

Given an objective distribution $p$ over $(\theta ,q,\phi )$, a type-$M$
consumer forms the following subjective belief over the latent add-on $q$
conditional on the observed price $\phi $ (as long as $\phi $ is realized
with positive probability under $p$):%
\begin{equation}
p_{M}(q\mid \phi )=\sum_{\theta _{M}}p(\theta _{M}\mid \phi )p(q\mid \theta
_{M})  \label{coarse beliefs}
\end{equation}

This formula represents a thought process that bears close resemblance to
Mailath and Samuelson's (2020) \textquotedblleft model-based
inference\textquotedblright . The consumer infers the exogenous variables in
his subjective model from observed prices, based on correct long-run
statistical data. He then uses this intermediate inference to predict the
add-on, again based on correct long-run data. His error is that he omits
exogenous variables that confound the relation between price and add-on. In
other words, the error can be described as \textquotedblleft confounder
neglect\textquotedblright\ (see Spiegler (2023)). In Section 5, we embed the
belief-formation model given by (\ref{coarse beliefs}) in a more general
formalism due to Spiegler (2016), in which consumers perceive market
regularities through the prism of a subjective causal model.

A key property of (\ref{coarse beliefs}) is that it is unbiased \textit{on
average} --- i.e.,%
\[
\sum_{\phi }p(\phi )p_{M}(q\mid \phi )\equiv p(q)
\]%
Thus, while the consumer may fail to draw correct add-on inferences from
prices, the forecasts are not systematically biased. This distinguishes our
model from a strand in the literature that includes Gabaix and Laibson
(2006) and Heidhues et al. (2016,2017), where consumers neglect hidden
charges altogether.

A consumer of cognitive type $M$ is active given the price $\phi $ if $v\geq
\phi +E_{M}(q\mid \phi )$, where $E_{M}(q\mid \phi )$ is the expected add-on
conditional on $\phi $ according to (\ref{coarse beliefs}). The demand
contributed by type-$M$ consumers is the measure of such consumers who
satisfy this inequality given $p$.\bigskip 

\noindent \textit{Equilibrium}

\noindent Consider a function $h$ from states $\theta $ to prices $\phi $.
This function, the objective distribution $\mu $ over states, and the
distribution of active firms given by (\ref{quality supply}), induce the
joint probability measure $p$ over $\theta ,\phi ,q$. In particular, $%
p(\theta )=\mu (\theta )$ and $p(\phi =h(\theta )\mid \theta )=1$ for every $%
\theta $. This is the objective distribution that type-$M$ consumers distort
into $p_{M}(q\mid \phi )$.

We say that $h$ is a \textit{competitive equilibrium} if for every pair $%
(\theta ,h(\theta ))$, total supply is equal to the total demand induced by
the distribution $p$ (which in turn is shaped by $h$). We say that a
competitive equilibrium is \textit{interior} if $\pi ^{\star }(\theta
,h(\theta ))\in (0,1)$ for every $\theta $ --- that is, there are positive
measures of both active and inactive firms in each state.

\subsection{Full Information Revelation}

A basic question in models of competitive markets with imperfectly informed
agents is whether equilibrium prices reveal the aggregate state $\theta $.
It turns out that interior equilibria in our model are fully
revealing.\bigskip

\begin{proposition}
\label{fulll reveal}In every interior equilibrium $h$, $\theta \neq \theta
^{\prime }$ implies $h(\theta )\neq h(\theta ^{\prime })$.
\end{proposition}

\begin{proof}
Consider an interior equilibrium $h$. Assume, contrary to the claim, that $%
h(\theta )=h(\theta ^{\prime })=\phi $ for some pair of states $\theta
,\theta ^{\prime }$. This means that consumers cannot distinguish between
the two states. As a result, the add-on forecast $E_{M}(q\mid \phi )$ is the
same in both states for every consumer type $M$. Consequently, aggregate
demand is the same in both states. Turning to the supply side, by assumption 
$S(\theta )\neq S(\theta ^{\prime })$. Therefore, the L.H.S of (\ref{active
condition}) is different in the two states, such that $\pi ^{\star }(\theta
,\phi )\neq \pi ^{\star }(\theta ^{\prime },\phi )$. Thus, supply is
different in the two states while the price is the same. This can only be
consistent with market clearing if demand is flat around $\phi $ in $\theta $
and $\theta ^{\prime }$. But since demand is downward-sloping around
interior-equilibrium prices, we obtain a contradiction.\bigskip
\end{proof}

This result means that in any interior equilibrium, a consumer with rational
expectations would perfectly deduce the state from the equilibrium price,
and therefore have a correct assessment of the expected add-on according to (%
\ref{quality}).

We say that a distribution $p$ over $(\theta ,\phi ,q)$ is \textit{fully
revealing} if both conditional distributions $(p(\phi \mid \theta ))$ and\ $%
(p(\theta \mid \phi ))$ are degenerate. In particular, when $p$ is fully
revealing, we use $\theta ^{p}(\phi )$ to denote the unique value of $\theta 
$ for which $p(\theta \mid \phi )=1$. This enables us to simplify (\ref%
{coarse beliefs}) into%
\begin{equation}
p_{M}(q\mid \phi )=\sum_{\theta ^{\prime }}\mu (\theta ^{\prime }\mid \theta
_{M}^{\prime }=\theta _{M}^{p}(\phi ))p(q\mid \theta ^{\prime })
\label{simple coarse formula}
\end{equation}%
where%
\[
\mu (\theta ^{\prime }\mid \theta _{M}^{\prime }=\theta _{M})=\frac{\mu
(\theta ^{\prime })}{\sum_{\theta ^{\prime \prime }\mid \theta _{M}^{\prime
\prime }=\theta _{M}}\mu (\theta ^{\prime \prime })} 
\]

Thus, the consumer forms his net willingness to pay for the product as if he
learned the realization of the state variables in his model. At the same
time, he fails to draw any inference from the event in which he trades with
firms --- which, in the $\varepsilon \rightarrow 0$ limit, is that he has
the highest net willingness to pay in the market. That is, the consumer
essentially commits a \textquotedblleft winner's curse\textquotedblright\
fallacy.

\subsection{Rational Expectations Benchmark}

Our model includes Rational Expectations Equilibrium (REE) as a special
case, when the consumer's type is $M=\{1,...n\}$ --- i.e., he does not
ignore any exogenous variable. In this case, $p_{M}(q\mid \phi )\equiv
p(q\mid \theta ^{p}(\phi ))$. The reason is that by Proposition \ref{fulll
reveal}, $\phi $ is a deterministic, one-to-one function of $\theta $ in
interior equilibrium.

Full revelation also means that we can analyze equilibria separately for
each state. Let us derive the equilibrium for the \textit{%
homogenous-preference} limit $\varepsilon \rightarrow 0$, where demand is
flat because all consumers have a net willingness to pay of 
\begin{equation}
v^{\ast }-\frac{1+\pi ^{\star }(\theta ,h(\theta ))}{2}S(\theta )
\label{REE price}
\end{equation}%
This is the expression for the equilibrium price $h(\theta )$ in the $%
\varepsilon \rightarrow 0$ limit, in terms of the threshold $\pi ^{\star
}(\theta ,h(\theta ))$. By definition, this threshold satisfies (\ref{active
condition}) bindingly in interior equilibrium when the market price is $%
h(\theta )$. Combining these equations, we obtain%
\begin{equation}
\pi ^{\star }(\theta ,h(\theta ))=1-\frac{2\Delta }{S(\theta )}
\label{RE eq}
\end{equation}%
It follows that an interior equilibrium exists whenever $2\Delta <S^{\min }$%
. Plugging (\ref{RE eq}) into (\ref{REE price}), the equilibrium price and
expected add-on level in state $\theta $ are%
\begin{eqnarray}
h(\theta ) &=&v^{\ast }+\Delta -S(\theta )  \label{REE price and add-on} \\
\bar{q}(\theta ) &=&S(\theta )-\Delta  \nonumber
\end{eqnarray}%
The total expected payment in state $\theta $ is $h(\theta )+\bar{q}(\theta
)=v^{\ast }$, such that consumers end up paying their net willingness to pay
for the product.

The interior REE is socially inefficient. Since $\Delta >0$ and the add-on
is a pure transfer, the efficient outcome is to maximize production ---
i.e., all firms should be active ($\pi ^{\star }=0$) in every state.
Interior equilibria violate this requirement, by a standard
adverse-selection argument. The state $\theta $ is an aggregate statistic
that determines the potential for hidden transfers in the market, yet firms
differ in their ability to realize this potential. Even when consumers
perfectly infer $\theta $ from the market price, the equilibrium involves
adverse selection because active firms are those with high ability to
generate the exploitative hidden transfer. This lowers consumers'
willingness to pay for the product, which in turn lowers the equilibrium
price and disincentives low-$\pi $ firms from entering. The REE volume of
trade is thus below the efficient level.

\section{Analysis}

This section is devoted to characterizing interior equilibrium in our model.
We take the following for granted throughout the section. First, we make use
of the result (Proposition \ref{fulll reveal}) that interior equilibria are
fully revealing. Second, we focus on the $\varepsilon \rightarrow 0$ limit,
where all consumers' bare valuation of the product is $v^{\ast }$. In any
equilibrium $h$ of this limit case,

\begin{eqnarray}
h(\theta ) &=&v^{\ast }-\min_{M\in \mathcal{M}}\int_{q}p_{M}(q\mid \phi
=h(\theta ))q  \label{eq_max_WTP} \\
&=&v^{\ast }-\min_{M\in \mathcal{M}}\sum_{\theta ^{\prime }}\mu (\theta
^{\prime }\mid \theta _{M}^{\prime }=\theta _{M})\bar{q}(\theta ^{\prime }) 
\nonumber
\end{eqnarray}%
for every state $\theta $ (the second equality makes use of (\ref{simple
coarse formula})). That is, the equilibrium price in each state is equal to
the highest net willingness to pay among all cognitive consumer types. The
types that trade with firms in $\theta $ are the ones with the lowest (most
optimistic) estimate of the exploitative add-on.

We will often make use of a simple relation between equilibrium prices and
add-on levels in each state:%
\begin{equation}
h(\theta )=S(\theta )+c-2\bar{q}(\theta )  \label{price and addon}
\end{equation}%
This equation follows from (\ref{active condition}) and (\ref{quality}),
when we plug $\phi =h(\theta )$ and make use of the fact that (\ref{active
condition}) is binding at $\pi ^{\ast }(\theta ,h(\theta ))$ in an interior
equilibrium. Equation (\ref{price and addon}) allows us to go back and forth
between statements about add-ons and statements about prices.

Consider the following restriction on the model's primitives:%
\begin{equation}
S^{\max }-S^{\min }<2\Delta <S^{\min }  \label{condition theorem}
\end{equation}%
The proof of the following result, as well as some of the later ones,
appears in the Appendix.\bigskip

\begin{proposition}
\label{characterization} Suppose that condition (\ref{condition theorem})
holds. Then, there exists a unique interior equilibrium. The expected
equilibrium add-on level in each state is given by the functional equation:%
\begin{equation}
\bar{q}(\theta )=\frac{1}{2}\left[ S(\theta )-\Delta +\min_{M\in \mathcal{M}%
}\sum_{\theta ^{\prime }}\mu (\theta ^{\prime }\mid \theta _{M}^{\prime
}=\theta _{M})\bar{q}(\theta ^{\prime })\right]  \label{bellman}
\end{equation}%
In particular, $S^{\min }-\Delta \leq \bar{q}(\theta )\leq S^{\max }-\Delta $
for every $\theta $.\bigskip
\end{proposition}

Equation (\ref{bellman}) has the exact form of a \textit{Bellman equation},
where the \textquotedblleft discount factor\textquotedblright\ is $\frac{1}{2%
}$; each action corresponds to one of the models in $\mathcal{M}$; and the
\textquotedblleft transition probability\textquotedblright\ from $\theta $
to $\theta ^{\prime }$ induced by $M$ is $\mu (\theta ^{\prime }\mid \theta
_{M}^{\prime }=\theta _{M})$. Thus, in interior equilibrium,
\textquotedblleft the market\textquotedblright\ acts\ as if it tries to
solve a Markov Decision Problem of minimizing a discounted sum of add-on
charges, where the transition probabilities are derived from consumers'
coarse beliefs.

The Bellman equation itself is an immediate consequence of putting the
supply and demand equations (\ref{eq_max_WTP}) and (\ref{price and addon})
together. The proof of Proposition \ref{characterization} is mostly devoted
to establishing that the solution of (\ref{bellman}) defines an interior
equilibrium. The bounds on $\bar{q}(\theta )$ are the REE add-on levels in
the states having extremal values of $S$, as given by (\ref{REE price and
add-on}).

Unlike REE, the equilibrium equations for different states are $not$
mutually independent. The reason is that consumers are imperfectly
discerning, hence their willingness to pay in one state can reflect the
expected add-on in other states. The bounds on the equilibrium levels of
expected add-ons mean that their range is \textit{more compressed}, relative
to REE.

Note that condition (\ref{condition theorem}) ensures the existence of
interior equilibrium for any $\mathcal{M}$ and $\mu $. In applications that
assume specific $\mathcal{M}$ and $\mu $, the condition can be relaxed. Note
also that since our definition of equilibrium focuses entirely on the price
function $h$, uniqueness of interior equilibrium does not extend to
allocations. In particular, if two consumer cognitive types happen to have
the same add-on forecast, we are agnostic about who trade is distributed
between these two types.

\subsection{An Illustrative Example with\ Two State Variables}

This sub-section presents an example that demonstrates the characterization
of interior equilibrium given by (\ref{bellman}). The example also shows
that consumers' equilibrium payoffs can be non-monotone with respect to a
natural measure of their sophistication.

Let $n=2$, $\mu =U\{(0,0),(0,1),(1,0)\}$, and $S(0,0)<S(0,1)\approx S(1,0)$.
The set of cognitive types $\mathcal{M}$ consists of all subsets of $\{1,2\}$%
. Thus, type $\{1,2\}$ has rational expectations; type $\emptyset $ has
fully coarse beliefs because he cannot perceive any correlation between
price and add-on; whereas types $\{1\}$ and $\{2\}$ have partially coarse
beliefs because they omit one variable from their subjective models.

We now guess an interior equilibrium. Type $\{1,2\}$ buys the product in
state $(0,0)$ (in which his belief assigns probability one to this state);
type $\{1\}$ buys the product in state $(0,1)$ (in which his belief is
uniform over $(0,0)$ and $(0,1)$); and type $\{2\}$ buys the product in
state $(1,0)$ (in which his belief is uniform over $(0,0)$ and $(1,0)$).
Type $\emptyset $ never buys the product. Under this guess,\ (\ref{bellman})
is reduced to the following system of linear equations:%
\begin{eqnarray}
2\bar{q}(0,0) &=&S(0,0)-\Delta +\bar{q}(0,0)  \label{bellman example} \\
2\bar{q}(0,1) &=&S(0,1)-\Delta +\frac{1}{2}\bar{q}(0,1)+\frac{1}{2}\bar{q}%
(0,0)  \nonumber \\
2\bar{q}(1,0) &=&S(1,0)-\Delta +\frac{1}{2}\bar{q}(1,0)+\frac{1}{2}\bar{q}%
(0,0)  \nonumber
\end{eqnarray}%
The solution is%
\begin{eqnarray}
\bar{q}(0,0) &=&-\Delta +S(0,0)\medskip   \label{example solution} \\
\bar{q}(0,1) &=&-\Delta +\frac{2S(0,1)+S(0,0)}{3}  \nonumber \\
\bar{q}(1,0) &=&-\Delta +\frac{2S(1,0)+S(0,0)}{3}  \nonumber
\end{eqnarray}%
In order for the solution to define an interior equilibrium, $\frac{1}{2}%
S(\theta )<\bar{q}(\theta )<S(\theta )$, which holds whenever $2\Delta
<S(0,0)$. Note that this is the condition for interior REE, which is more
lenient than (\ref{condition theorem}).

The following table summarizes the subjective add-on estimates $E_{M}(q\mid
\theta )$ for every type $M$ (we use the abbreviated notation $q_{\theta
_{1}\theta _{2}}$ for $\bar{q}(\theta _{1},\theta _{2})$):

\[
\begin{array}{cccc}
Type\backslash State & 0,0 & 0,1 & 1,0 \\ 
\{1,2\} & q_{00} & q_{01} & q_{10} \\ 
\{1\} & \frac{1}{2}(q_{00}+q_{01}) & \frac{1}{2}(q_{00}+S_{01}) & q_{10} \\ 
\{2\} & \frac{1}{2}(q_{00}+q_{10}) & q_{01} & \frac{1}{2}(q_{00}+q_{10}) \\ 
\emptyset & \frac{1}{3}(q_{00}+q_{01}+q_{10}) & \frac{1}{3}%
(q_{00}+q_{01}+q_{10}) & \frac{1}{3}(q_{00}+q_{01}+q_{10})%
\end{array}%
\]%
Since $\bar{q}(0,0)<\bar{q}(0,1)\approx \bar{q}(1,0)$, this table confirms
our guess of the types with the lowest add-on estimate in each state.

In the interior equilibrium we derived, the partially coarse types $\{1\}$
and $\{2\}$ earn negative payoffs in the states in which they buy the
product, as their net willingness to pay exceeds the rational type's in
these states. In contrast, the fully coarse type $\emptyset $, who is
intuitively less sophisticated than the partially coarse types, enjoys a
\textquotedblleft loser's blessing\textquotedblright : He earns zero payoffs
because he never trades. Thus, the types who suffer a welfare loss are
sophisticated enough to infer from the observed price that one state
variable is favorable, but not sophisticated enough to understand that them
buying the product implies that the other state variable is unfavorable. As
a result, they underestimate the add-on and overpay for the product. This
kind of non-monotonicity in consumer sophistication has been observed in
previous works (most relatedly, by Ettinger and Jehiel (2011) and Eyster and
Piccione (2013)).

\subsection{Characterization Results}

In this sub-section we put Proposition \ref{characterization} to work. Our
first result examines how the interior equilibrium changes when we expand
the set of cognitive types $\mathcal{M}$ --- i.e., when consumers become
more diverse in terms of their subjective models. The \textquotedblleft
Bellman\textquotedblright\ characterization of interior equilibrium means
that expanding $\mathcal{M}$ is formally equivalent to expanding the set of
actions in a Markov Decision Problem (MDP). This equivalence enables us to
tap into standard results on solutions of MDPs and apply them to the present
context, where they have very different meaning.\bigskip

\begin{proposition}
\label{prop adding types}Suppose condition (\ref{condition theorem}) holds.
Then, adding a new type $M$ to $\mathcal{M}$ has the following effects on
the unique interior equilibrium:\medskip \newline
$(i)$ $\bar{q}(\theta )$ weakly decreases for every state $\theta $.\medskip 
\newline
$(ii)$ $h(\theta )$ weakly increases for every $\theta $.\medskip \newline
$(iii)$ Social surplus weakly increases in each state.\medskip \newline
$(iv)$ Individual consumers' net payoff weakly decreases in each state.
Moreover, if $\mathcal{M}$ includes a rational type, then consumers incur an
ex-ante welfare loss, which increases when $M$ is added to $\mathcal{M}$.
\end{proposition}

\begin{proof}
By Proposition \ref{characterization}, $\bar{q}(\theta )$ is formally the
solution to a finite-state MDP of minimizing a discounted expected cost
function, where $\mathcal{M}$ is the set of feasible actions in this MDP.
Expanding the set of feasible actions weakly improves the value function at
each state, which implies $(i)$. Property $(ii)$ then immediately follows
from $(i)$ and equation (\ref{price and addon}).

To see why property $(iii)$ follows from $(i)$, note that by (\ref{quality}%
), $\bar{q}(\theta )$ decreases if and only if $\pi ^{\ast }(\theta )$
decreases. Therefore, the expansion of $\mathcal{M}$ leads to a weak
decrease in $\pi ^{\ast }(\theta )$ in each state $\theta $. This means that
there are more active firms --- and hence more trade --- in each state. As
we saw, in this model social welfare is pinned down by the volume of trade.

As to property $(iv)$, note that consumers who do not trade in a given state
earn zero payoffs. Consumers who do trade in a state $\theta $ earn a net
payoff of $v^{\ast }-h(\theta )-\bar{q}(\theta )$. Plugging (\ref{price and
addon}), this expression becomes $\Delta -S(\theta )+\bar{q}(\theta )$.
Since the expansion of $\mathcal{M}$ leads to a weak decrease in $\bar{q}%
(\theta )$, active consumers' net payoff in $\theta $ weakly decreases, too.

When $\mathcal{M}$ includes a rational type, the net payoff of any consumer
who trades in any state must be weakly negative, because the equilibrium
price is equal to this type's willingness to pay and therefore lies weakly
above the rational-expectations willingness to pay. As we saw above, the
volume of trade --- which is equal to the measure of consumers who trade ---
weakly increases in each state when we expand $\mathcal{M}$. Thus, not only
does the net payoff loss of each trading consumer weakly increases when we
expand $\mathcal{M}$, but there are also weakly more consumers who trade in
each state. This means that consumers' ex-ante welfare loss weakly goes
up.\bigskip
\end{proof}

Thus, expanding the set of cognitive types shifts payments from hidden
add-ons to salient prices --- i.e., add-ons decrease while prices increase.
The net effect on individual consumers' welfare is negative. The intuition
is that the expansion of $\mathcal{M}$ leads to an increase in demand, and
therefore higher equilibrium prices in each state. In response, the pool of
active firms becomes less adversely selective, as lower-$\pi $ types enter
the market thanks to the higher price. This in turn means that latent
exploitation shrinks in equilibrium. The net effect on trading consumers'
net payoff is negative, as they pay more in total.

Recall that in REE, $\bar{q}(\theta )=S(\theta )-\Delta $ for every $\theta $%
. Therefore, the ex-ante expected add-on in REE is $\bar{S}-\Delta $. The
following result draws on Proposition \ref{prop adding types} to show that
the ex-ante expected add-on in interior equilibrium is weakly below this REE
level. The result also shows that the lowest possible expected add-on given $%
\bar{S}$ is sustainable in equilibrium (for a suitable specification of
primitives).\bigskip

\begin{proposition}
\label{prop average S} In interior equilibrium, the expected add-on is in $[%
\frac{1}{2}\bar{S},\bar{S}-\Delta ]$. Moreover, the lower bound can be
approximated arbitrarily well by interior equilibrium for a suitable
selection of $\Theta ,S,\mu ,\mathcal{M}$ that is compatible with $\bar{S}$%
.\bigskip
\end{proposition}

The argument behind the proposition's first part is simple. When $\mathcal{M}
$ is a singleton, (\ref{bellman}) becomes a linear equation in $\bar{q}%
(\theta )$, for every $\theta $. This linearity, coupled with the
unbiasedness-on-average property of consumers' beliefs, implies that the
ex-ante expected add-on in interior equilibrium coincides with the REE
level. When we add cognitive types, Proposition \ref{prop adding types}
implies a drop in the expected add-on.

The lower bound on the expected add-on is attained in a large-$n$ variant on
the example of Section 3.1. In equilibrium, the trading consumer in every
state has an optimistic belief in the sense that he believes that the
expected add-on hits (exactly or approximately) its lowest possible level
(i.e., $S=S^{\min }$ and $\pi ^{\ast }=0$). The equilibrium outcome is
nearly efficient, as $\pi ^{\ast }\approx 0$ in every state, such that there
is no adverse selection, and the equilibrium headline price is close to $c$
in every state.

Thanks to (\ref{price and addon}), Proposition \ref{prop average S} has an
immediate implication for equilibrium headline prices.\bigskip

\begin{corollary}
The ex-ante expected price in interior equilibrium is weakly above its REE
level $v^{\ast }+\Delta -\bar{S}$. \bigskip
\end{corollary}

Turning from average price components to their range, recall that by
Proposition \ref{characterization}, the range of expected add-ons in
interior equilibrium is compressed relative to REE. The following result
obtains a mirror image for prices, as long as there are rational consumers
in the market.\bigskip

\begin{proposition}
\label{proposition range}if $\mathcal{M}$ includes a rational type, then $%
2v^{\ast }-c-S^{\max }\leq h(\theta )\leq 2v^{\ast }-c-S^{\min }$ for every $%
\theta $ in the interior equilibrium. Moreover, the R.H.S inequality is
binding when $S(\theta )=S^{\min }$.\bigskip
\end{proposition}

Thus, adding imperfectly discerning consumers to a market that contains
rational consumers widens equilibrium price fluctuations. To see why the
presence of rational consumers is required for this result, suppose $%
\mathcal{M}$ consists of a single type $M=\emptyset $. This type has fully
coarse beliefs, and therefore his willingness to pay is constant across
states. It follows that the equilibrium price is absolutely rigid. Moreover,
by (\ref{price and addon}), it coincides with the expected REE price. In
this case, equilibrium prices are obviously compressed relative to REE.

\section{Mutually Beneficial Add-Ons}

So far, we have assumed that latent add-ons are purely exploitative, namely
a transfer from consumers to firms. In many real-life contexts, however,
add-on features generate surplus for both parties. For instance, the add-on
can be a follow-up service which, due to compatibility issues, the consumer
can only get from whoever sold him the basic product. If demand for this
service is linearly downward-sloping, the optimal monopoly price for the
service will split the surplus equally between the consumer and the firm.

In this section, we\ present a variant on our model that covers such cases.
Assume that when a consumer buys from a type-$\pi $ firm in state $\theta $, 
\textit{each of them} obtains a latent payoff of $\pi S(\theta )$. We refer
to $q=\pi S(\theta )$ as the \textit{quality} that the consumer gets in this
case, and to $\bar{q}(\theta )$ (as defined by (\ref{quality})) as the 
\textit{average quality} in state $\theta $. As we will see, since
consumers' latent payoff is positive, interior equilibrium will require us
to assume that $\Delta =v^{\ast }-c<0$ --- namely, the basic product
generates a \textit{negative} surplus. All the other modeling assumptions
and definitions remain as in the basic model of Section 2. In particular,
the supply side behaves exactly as in the basic model, and interior
equilibrium continues to be fully revealing.

We now focus on the $\varepsilon \rightarrow 0$ limit, where demand is
nearly homogeneous. In any equilibrium $h$ of this limit case, 
\[
h(\theta )=v^{\star }+max_{M\in \mathcal{M}}\sum_{\theta ^{\prime }}\mu
(\theta ^{\prime }\mid \theta _{M}^{\prime }=\theta _{M})\bar{q}(\theta
^{\prime })
\]%
for every state $\theta $. Compare this expression with (\ref{eq_max_WTP}).
The equilibrium price in state $\theta $ is determined by the consumer type
with the \textit{highest} add-on estimate in that state (whereas in the
basic model, the type with the \textit{lowest} estimate determined the
price). Combining this equation for $h(\theta )$ with the supply-driven
equation (\ref{price and addon}), we obtain%
\begin{equation}
\bar{q}(\theta )=\frac{1}{2}\left[ S(\theta )-\Delta -max_{M\in \mathcal{M}%
}\sum_{\theta ^{\prime }}\mu (\theta ^{\prime }\mid \theta _{M}^{\prime
}=\theta _{M})\bar{q}(\theta ^{\prime })\right]   \label{quasi_bel}
\end{equation}

This equation is exactly the same as (\ref{bellman}), except for the minus
sign before the third term inside the brackets. In other words, it is like a
Bellman equation with a negative discount factor. The equation defines a
contraction mapping, and so it has a unique solution, pinning down $h(\theta
)$ and $\pi ^{\star }(\theta ,h(\theta ))$. To guarantee that the
equilibrium is indeed interior, we impose the following condition on the
primitives:%
\begin{equation}
-\frac{2}{3}\Delta <S^{min}<S^{max}<-\Delta  \label{condition beneficial}
\end{equation}

While it is tempting to think that (\ref{quasi_bel}) can be used to recover
all of the results from Section 3 (possibly with a change of sign), the next
example illustrates that this is not the case. Specifically, the example
shows that expanding the set of cognitive types need not have a uniform
effect on equilibrium add-on levels across states (unlike Proposition \ref%
{prop adding types}).

Let $n=1$, $\theta \in \{0,1\}$, and assume $\mu $ is uniform. Let $%
S(0)=k<1=S(1)$ and assume (\ref{condition beneficial}) holds. Suppose $%
\mathcal{M}$ consists of a single, \textquotedblleft fully
coarse\textquotedblright\ type $M=\emptyset $. This type's add-on estimate
is $(\bar{q}(0)+\bar{q}(1))/2$ in both states. The solution to (\ref%
{quasi_bel}) is $\bar{q}(0)=(5k-4d-1)/12$ and $\bar{q}(1)=(5-4d-k)/12$. Now
add a rational type to $\mathcal{M}$. We can guess and verify that in
equilibrium, the rational type buys the product in $\theta =1$ and the
coarse type buys the product in $\theta =0$. The solution to (\ref{quasi_bel}%
) is $\bar{q}(0)=(6k-5d-1)/15$ and $\bar{q}(1)=(1-d)/3$. Thus, as a result
of the expansion of $\mathcal{M}$, expected equilibrium quality decreases in 
$\theta =1$ and increases in $\theta =0$.

Recall that in the basic model, where add-ons are exploitative, the expected
equilibrium add-on level is below its REE level --- i.e., there is a shift
from latent to salient price components. The same holds in the present
variant, as long as there are rational consumers in the market.\bigskip

\begin{proposition}
\label{mutually_beneficial_addition} Suppose that $\mathcal{M}$ includes a
rational type. Then, in the interior equilibrium, for every state $\theta $, 
$\bar{q}(\theta )$ is weakly below its REE level and $h(\theta )$ is weakly
above its REE level.
\end{proposition}

\begin{proof}
Denote $c(\theta )=\frac{1}{3}(S(\theta )-\Delta )$. It is possible to
rewrite (\ref{quasi_bel}) as 
\[
\frac{2}{3}\bar{q}(\theta )+\frac{1}{3}max_{M\in \mathcal{M}}\sum_{\theta
^{\prime }}\mu (\theta ^{\prime }\mid \theta _{M}^{\prime }=\theta _{M})\bar{%
q}(\theta ^{\prime })=c(\theta )
\]%
It follows that in REE, $\bar{q}(\theta )=c(\theta )$ in every state $\theta 
$. Since $\mathcal{M}$ includes a rational consumer type, $\bar{q}(\theta
)\leq c(\theta )$ in every state $\theta $. The weakly lower expected
quality implies that the fraction of active firms in the market is weakly
higher in each state than in REE (i.e., $\pi ^{\ast }$ is lower), and so $%
h(\theta )$ must be weakly higher in each state.\bigskip 
\end{proof}

This effect is the same as in the basic model, although it now requires us
to assume that $\mathcal{M}$ includes a rational type. The effect's welfare
implications, however, are very different from what we observed in the basic
model.\bigskip 

\begin{proposition}
\label{mutually_beneficial_welfare} Suppose that $\mathcal{M}$ includes a
rational type. Then, when $\mathcal{M}$ is expanded, equilibrium social
surplus weakly decreases.\bigskip
\end{proposition}

When the add-on is mutually beneficial, the competitive market is \textit{%
positively selective} --- i.e., the firm types that enter the market are the
ones that create more latent surplus for consumers. (By comparison, the
market in our basic model exhibits adverse selection.) In REE, consumers
earn zero net payoffs on average, which means that trading with the marginal
firm type $\pi ^{\ast }$ is harmful for consumers. Since this firm type is
indifferent to market entry, trading with it is socially harmful. In other
words, the REE volume of trade is excessive from the perspective of social
welfare. When $\mathcal{M}$ includes a rational type and we expand this set,
even lower-quality firms enter the market, which exacerbates this social
harm.

In summary, competitive markets with diversely discerning markets function
differently when latent product features are mutually beneficial and when
they are exploitative. Technically, the difference finds expression in the
sign of the Bellman-like equation that characterizes interior equilibrium.
Economically, the difference is that markets with mutually beneficial latent
add-ons are positively selective, whereas markets with exploitative latent
add-ons are adversely selective.

\section{Beliefs Based on Subjective Causal Models}

In this section we revert to the exploitative-add-on version of the model,
and extend the consumer belief-formation model presented in Section 2. The
more general model, based on Spiegler (2016), assumes that every cognitive
type represents a \textit{subjective causal model} that postulates
qualitative causal links among several variables: The observed price $\phi $%
, the add-on $q$, and some of the state variables $\theta _{1},...,\theta
_{n}$. This extension will enable us to capture varieties of partially
discerning consumers beyond the basic model's scope. In turn, this will give
rise to novel supply and demand responses to external shocks. As we will
see, \textit{all} the results in previous sections will extend to this more
general model.

A causal model is a \textit{directed acyclic graph} (DAG) $G=(N,R)$, where $%
N $ is a set of nodes and $R$ is a set of directed links. Each node in $N$
represents a variable, and a link in $R$ represents a perceived causal
relation between two variables. Let $\mathcal{G}$ be the set of subjective
causal models in the consumer population. This is the analogue of $\mathcal{M%
}$ in the basic model. As before, we assume that the measure of consumers of
each of these types is greater than $1$.

We impose the following restrictions on every $G\in \mathcal{G}$. First, it
must include nodes that represent $\phi $ and $q$ (because the consumer
tries to infer the add-on from the headline price). Second, it does not have
links of the form $\phi \rightarrow \theta _{i}$ or $q\rightarrow \theta
_{i} $. This restriction means that consumers realize that state variables
are exogenous whereas price and add-on are endogenous.

It is sometimes helpful to label causal-model variables as $(x_{i})_{i\in N}$%
. Abusing notation, let $R(i)$ be the set of nodes that send a directed link
into $i$. A node $i$ is $ancestral$ if $R(i)=\emptyset $. When the objective
joint distribution over all variables is $p$, a consumer whose subjective
DAG is $G=(N,R)$ forms the following subjective probabilistic belief over
the variables in his model:%
\begin{equation}
p_{G}(x_{N})=\dprod\limits_{i\in N}p(x_{i}\mid x_{R(i)})  \label{factor}
\end{equation}%
This is a standard Bayesian-network factorization formula (see Pearl (2009)
and Spiegler (2016)).

Our analysis will focus on the following subclass of DAGs.\bigskip

\begin{definition}
A DAG $G=(N,R)$ is perfect if, for every triple of nodes $i,j,k\in N$, $%
i,j\in R(k)$ implies $i\in R(j)$ or $j\in R(i)$.\bigskip
\end{definition}

\noindent In a perfect DAG, the parents of every node form a clique. The
basic model of Section 2 is a special case of the perfect-DAG formalism. The
set $M$ is a subset of the nodes that represent $\theta $. All the nodes in $%
M$ are mutually linked. In addition, $R(\phi )=R(q)=M$.

There are two motivations for adopting the perfect-DAG formalism. First,
perfect DAGs subsume earlier equilibrium market models with non-rational
expectations as special cases (including the basic model of Section 2),
while making room for new ones. Second, perfect DAGs satisfy the
unbiasedness-on-average property that we observed in Section 2.\bigskip

\begin{remark}
\label{perfect marginal} Suppose $G$ is a perfect DAG. Then, for every $p$
that arises from an interior equilibrium $h$,%
\begin{equation}
\sum_{\theta }\mu (\theta )p_{G}(q\mid h(\theta ))\equiv \sum_{\phi }p(\phi
)p_{G}(q\mid \phi )\equiv p(q)  \label{no distortion of average quality}
\end{equation}
\end{remark}

\noindent The left-hand identity arises from $h$ being a one-to-one function
of $\theta $ in interior equilibrium. For the right-hand identity, see
Spiegler (2020a,b).

To illustrate the perfect-DAG formalism, let $G_{ch}:\phi \leftarrow \theta
_{1}\rightarrow \theta _{2}\rightarrow q$. This DAG represents a causal
model that postulates $\theta _{1}$ as the sole direct cause of $\phi $ and $%
\theta _{2}$, and $\theta _{2}$ as the sole direct cause of $q$. It captures
consumers who mistakenly think that different external factors affect the
product's salient and latent components, whereas in reality both components
are jointly determined by all state variables. This DAG induces the
subjective belief%
\[
p_{G_{ch}}(\theta _{1},\theta _{2},q,\phi )=p(\theta _{1})p(\theta _{2}\mid
\theta _{1})p(\phi \mid \theta _{1})p(q\mid \theta _{2}) 
\]%
which in turn yields the conditional belief%
\[
p_{G}(q\mid \phi )=\sum_{\theta _{1},\theta _{2}}p(\theta _{1}\mid \phi
)p(\theta _{2}\mid \theta _{1})p(q\mid \theta _{2}) 
\]%
Because $\mu $ has full support over $\theta $, this expression is
well-defined.

\subsection{Generalizing the Bellman\ Equation}

We now present a lemma that provides a convenient characterization of the
conditional belief $p_{G}(q\mid \phi =h(\theta ))$ when $G$ is a perfect
DAG. In what follows, we refer to a system of conditional probabilities $%
\beta =(\beta (\theta ^{\prime }\mid \theta ))_{\theta ,\theta ^{\prime }\in
\Theta }$ as a \textit{transition matrix}. Recall that $\theta ^{p}(\phi )$
is the state $\theta $ that generates the price $\phi $ in a fully revealing 
$p$.\bigskip

\begin{lemma}
\label{lemma beta}Fix a distribution $\mu $ and a perfect DAG $G=(N,R)$.
Then, there exists a unique transition matrix $\beta $ satisfying the
following: For every fully revealing distribution $p$ over $(\theta ,\phi
,q) $ whose marginal over $\theta $ is $\mu $, and for every price $\phi $
in the support of $p$,%
\begin{equation}
p_{G}(q\mid \phi )=\sum_{\theta ^{\prime }}\beta (\theta ^{\prime }\mid
\theta ^{p}(\phi ))p(q\mid \theta ^{\prime })  \label{beta representation}
\end{equation}%
Moreover, $\mu $ is an invariant distribution of $\beta _{G}$.\bigskip
\end{lemma}

Thus, given $\mu $ and $G$, we have a simple representation of the
consumer's belief over the add-on conditional on the market price.\footnote{%
This representation is somewhat reminiscent of a model of misperception of
correlations by Ellis and Piccione (2017).} Instead of correctly inferring
the add-on distribution (\ref{quality supply}) in the state revealed by the
market price, the consumer effectively calculates a weighted average of the
add-on distributions associated with various \textquotedblleft
virtual\textquotedblright\ states; the weights on virtual states may vary
with the actual state. This representation is made possible by the property
that $p$ is fully revealing, such that there is a one-to-one mapping between
prices and states.

In the basic model of Section 2, $\beta (\theta ^{\prime }\mid \theta )=\mu
(\theta ^{\prime }\mid \theta _{M}^{\prime }=\theta _{M})$. For the DAG $%
G_{ch}:\phi \leftarrow \theta _{1}\rightarrow \theta _{2}\rightarrow q$
introduced above, 
\[
\beta (\theta _{1}^{\prime },\theta _{2}^{\prime }\mid \theta _{1},\theta
_{2})\equiv \mu (\theta _{2}^{\prime }\mid \theta _{1})\mu (\theta
_{1}^{\prime }\mid \theta _{2}^{\prime })
\]%
To illustrate this formula, let $n=2$, $\theta _{1},\theta _{2}\in \{0,1\}$,
and $\mu =U\{(0,0),(1,0),(0,1)\}$. Then, $\beta (0,0\mid 0,\cdot )=\beta
(1,0\mid 0,\cdot )=0.25$; $\beta (0,1\mid 0,\cdot )=0.5$; and $\beta
(0,0\mid 1,0)=\beta (1,0\mid 1,0)=0.5$. Observe that the transition matrix
assigns positive weight to $\theta _{1}^{\prime }\neq \theta _{1}$, even
though the consumer correctly infers $\theta _{1}$ from $\phi $. Moreover, $%
\beta (0,0\mid 1,0)>\beta (0,0\mid 0,\cdot )$.

Although the representation (\ref{beta representation}) is convenient,
treating it as a primitive would be inappropriate. First, $\beta $ is often
hard to interpret, whereas its DAG-based foundation is interpretable.
Second, recall that (\ref{beta representation}) takes $\mu $ as fixed. In
the absence of a deeper foundation for $\beta $, we have no guide for how to
modify it when $\mu $ changes.

A fully connected DAG (i.e., one in which every pair of nodes is linked)
that includes all $\theta $ variables induces rational expectations, because
in this case (\ref{factor}) becomes the standard chain rule for probability
distributions over $(\theta ,\phi ,q)$. However, this is not the only class
of perfect DAGs that are guaranteed to induce correct equilibrium beliefs,
because in equilibrium, $\phi $ is a deterministic function of $\theta $.
When a DAG $G$ does not exclude any of the $\theta $ variables, and every
pair of nodes is linked (except possibly $(\phi ,q)$), then it represents a
rational consumer. Likewise, a perfect DAG in which $\phi $ and $q$ are
directly linked induces rational expectations. The transition matrix that
represents such consumers is the unit matrix, $\beta (\theta \mid \theta )=1$
for all $\theta $.

Proposition \ref{characterization} extends to the present belief-formation
model whenever $\mathcal{G}$ is a collection of perfect DAGs. The
Bellman-like equation (\ref{bellman}) is modified into%
\begin{equation}
\bar{q}(\theta )=\frac{1}{2}\left[ S(\theta )-\Delta +\min_{G\in \mathcal{G}%
}\sum_{\theta ^{\prime }}\beta _{G}(\theta ^{\prime }\mid \theta )\bar{q}%
(\theta ^{\prime })\right]   \label{bellman dag}
\end{equation}%
where $\beta _{G}$ is the transition matrix that represents the perfect DAG $%
G$. Condition (\ref{condition theorem}) continues to ensure existence and
uniqueness of interior equilibrium. All the other results in Section 3
extend as well. The mutually beneficial add-on variant of Section 4 is
extended in the same manner. The quasi-Bellman equation that characterizes
interior equilibrium is the same as (\ref{bellman dag}), except that the
last term in the squared brackets is preceded by a minus sign (and the
condition for interior equilibrium is (\ref{condition beneficial})).

\subsection{The Two-State-Variables Example Revisited}

To illustrate the use of (\ref{bellman dag}) to characterize interior
equilibrium in the DAG-based extension, revisit the example of Section 3.1,
where $n=2$, $\theta _{1},\theta _{2}\in \{0,1\}$, $\mu
=U\{(0,0),(0,1),(1,0)\}$, and $S(0,0)<S(1,0)\approx S(0,1)$.

The set of cognitive types $\mathcal{G}$ consists of a rational type, and
the two chain DAGs $G_{1}:\phi \leftarrow \theta _{1}\rightarrow \theta
_{2}\rightarrow q$ and $G_{2}:\phi \leftarrow \theta _{2}\rightarrow \theta
_{1}\rightarrow q$. We presented the transition matrix that represents $%
G_{1} $ in the previous sub-section. The matrix that represents $G_{2}$ is: $%
\beta (0,0\mid \cdot ,0)=\beta (0,1\mid \cdot ,0)=0.25$, $\beta (1,0\mid
\cdot ,0)=0.5$, and $\beta (0,0\mid \cdot ,1)=\beta (0,1\mid \cdot ,1)=0.5$.
Note that $\beta (0,0\mid \cdot ,1)>\beta (0,0\mid \cdot ,0)$. Thus, each of
the chain-DAG types draws an optimistic inference about the add-on when the
state variable he directly infers from the price takes the \textquotedblleft
bad\textquotedblright\ value $1$, rather than the \textquotedblleft
good\textquotedblright\ value $0$.

We now guess an equilibrium, and later verify that our guess is indeed an
equilibrium. As before, the guess-and-verify method is valid because there
is at most one interior equilibrium. Suppose the rational type buys the
product in state $(0,0)$; type $G_{1}$ buys the product in state $(1,0)$;
and type $G_{2}$ buys the product in state $(0,1)$. Under this guess, (\ref%
{bellman dag}) takes the exact same form as (\ref{bellman example}), leading
to the same solution (\ref{example solution}) for $\bar{q}(\theta )$. Let us
verify that the type who buys in each state indeed has the lowest add-on
estimate. The following table presents expressions for each type's estimate
in each state (we use the abbreviated notation $q_{\theta _{1}\theta _{2}}$
for $\bar{q}(\theta )$):

\[
\begin{array}{cccc}
Type\backslash State & 0,0 & 0,1 & 1,0 \\ 
rational & q_{00} & q_{01} & q_{10} \\ 
G_{1} & \frac{1}{4}(q_{00}+q_{10})+\frac{1}{2}q_{01} & \frac{1}{4}%
(q_{00}+q_{10})+\frac{1}{2}q_{01} & \frac{1}{2}(q_{00}+q_{10}) \\ 
G_{2} & \frac{1}{4}(q_{00}+q_{01})+\frac{1}{2}q_{10} & \frac{1}{2}%
(q_{00}+q_{01}) & \frac{1}{4}(q_{00}+q_{01})+\frac{1}{2}q_{10}%
\end{array}%
\]%
Recall that (\ref{example solution}) implies $q_{00}<q_{01}\approx q_{10}$,
hence our guess is confirmed.

While the expected add-on in each state is the same as in Section 3.1, the
inference behind the trading consumer types' add-on estimates is different.
For example, when the state is $(1,0)$, type $G_{1}$ correctly infers $%
\theta _{1}=1$ from the equilibrium price. While this realization by itself
is associated with a high add-on (because the only state in which $\theta
_{1}=1$ is $(1,0)$), the type's DAG leads him to assign probability $\frac{1%
}{2}$ to the state $(0,0)$, in which the add-on is at its lowest. Thus,
unlike the example in Section 3.1, a pessimistic inference about the state
variable the consumer regards as the direct cause of prices leads to an
optimistic add-on forecast.

\subsection{\textquotedblleft Anomalous\textquotedblright\ Market
Fluctuations}

In competitive markets, fluctuations in prices and allocations reflect
supply and demand responses to external shocks. When REE fully reveals all
payoff-relevant information, these responses are as if the information is
public. In this sub-section, we demonstrate that under the DAG-based
extension of our model, equilibrium supply and demand responses to shocks
exhibit patterns that are impossible in REE (or under the basic model).
First, we show that although supply and demand shocks in our model are
perfectly correlated (negatively for most of the paper, positively in the
variant of Section 4), the supply and demand responses can be nearly
independent. Second, we extend the model by endowing consumers with private
information, and show that even though equilibrium prices fully reveal all
payoff-relevant aspects of the state, they can also respond to fluctuations
in consumers' private information. The common theme in both sub-sections is
that markets with imperfectly discerning consumers are more volatile
relative to REE.

\subsubsection{How Supply and Demand Co-Move}

The state $\theta $ in our model determines a zero-sum transfer from
consumers to firms. Therefore, under rational expectations, supply and
demand move in opposite directions in response to fluctuations in $\theta $.
This is evident from equations (\ref{REE price and add-on}), which
characterize REE. (In the variant of Section 4, shocks are of a
\textquotedblleft common value\textquotedblright\ nature, hence supply and
demand would move in tandem in response to shocks under rational
expectations.)

Now suppose $\mathcal{G}$ consists of a single \textquotedblleft fully
coarse\textquotedblright\ consumer, who does not perceive any correlation
between price and quality. This consumer will exhibit an absolutely rigid
demand, such that equilibrium price fluctuations only reflect supply
responses to shocks.

However, our model can also generate virtually independent supply and demand
movements in response. For illustration, let $\theta =(\theta _{1},\theta
_{2},\theta _{3})$, $\theta _{i}\in \{0,1\}$ for every $i$. Assume that $%
S(\theta )=\alpha _{1}\theta _{1}+\alpha _{2}\theta _{2}+\alpha _{3}\theta
_{3}+b$, where $b>0$ is a constant; and the weights $\alpha _{i}$ are all
positive and different from each other. Moreover, let $\alpha _{1},\alpha
_{3}\approx 0$, whereas $\alpha _{2}$ is bounded away from zero, such that
the maximal feasible add-on is almost entirely a function of $\theta _{2}$.
Assume that $\mu $ satisfies the following properties: $\theta _{1}$ and $%
\theta _{2}$ are statistically independent, and $\theta _{3}$ is some
function of these two state variables. Under this specification, the supply
function mainly responds to fluctuations in $\theta _{2}$, and exhibits
virtually no response to the other state variables conditional on $\theta
_{2}$.

Finally, assume $\mathcal{G}$ consists of a DAG $G:\phi \leftarrow \theta
_{1}\rightarrow \theta _{3}\rightarrow \theta _{2}\rightarrow q$. Even
though $\theta _{1}$ and $\theta _{2}$ are objectively independent, they may
be correlated according to the subjective belief $p_{G}$, as long as both $%
\theta _{1}$ and $\theta _{2}$ are correlated with $\theta _{3}$, since%
\[
p_{G}(\theta _{2}\mid \theta _{1})=\sum_{\theta _{3}^{\prime }}p(\theta
_{3}^{\prime }\mid \theta _{1})p(\theta _{2}\mid \theta _{3}^{\prime }) 
\]%
Eliaz et al. (2021) showed that this spurious subjective correlation can be
quite large.\footnote{%
Unlike other examples in this paper, $G$ displays a misunderstanding of the
statistical behavior of exogenous variables, in addition to the
misperception of how edogenous variables vary with them.}

In our context, what this observation means is that consumer demand will be
highly responsive to prices, because the consumer correctly infers $\theta
_{1}$ from the equilibrium price while exaggerating the correlation between $%
\theta _{1}$ and $q$ (as a result of the erroneous perception that $\theta
_{1}$ and $\theta _{2}$ are correlated). Thus, while supply will be almost
entirely a function of $\theta _{2}$, demand will be a function of $\theta
_{1}$. Since these two state variables are objectively independent, supply
and demand responses to external shocks will be virtually orthogonal. This
pattern of fluctuations is impossible in our basic model (which subsumes REE
as a special case).

\subsubsection{Partially Informed Consumers}

So far, we have assumed that consumers have no\ information about the state
(other than what they can learn from prices). Since equilibrium prices are
fully revealing, this lack of information is irrelevant if consumers have
rational expectations. We will now see that this irrelevance no longer holds
when consumers are imperfectly discerning.

To explore the role of partial consumer information, extend the model as
follows. For every consumer type $G$, there is a distinct variable $w_{G}$
which represents a noisy private signal of $\theta $ that type-$G$ consumers
observe. These consumers admit $w_{G}$ as a variable in their causal model,
such that $R(w_{G})$ is contained in the set of nodes that represent $\theta 
$, and $w_{G}$ itself is not a parent of any other node. Thus, the consumer
understands that $w_{G}$ is merely a signal of the exogenous state
variables, and therefore not a (direct or indirect) cause of any other
variable. For instance, $G$ can be%
\begin{equation}
\begin{array}{lllllll}
&  & w_{G} &  &  &  &  \\ 
&  & \uparrow & \nwarrow &  &  &  \\ 
\phi & \leftarrow & \theta _{1} & \rightarrow & \theta _{2} & \rightarrow & q%
\end{array}
\label{extended short chain}
\end{equation}

Extend $\mu $ to be a joint distribution over $\theta $ and $w=(w_{G})_{G\in 
\mathcal{G}}$. The induced joint distribution $p$ over all variables is
defined as usual. Thus, when the market price is $\phi $, a type-$G$
consumer uses the conditional subjective belief $p_{G}(q\mid \phi ,w_{G})$
to predict the add-on. When $G$ is given by (\ref{extended short chain}), we
can see that the consumer infers $\theta _{1}$ from the market price $\phi $%
, and then uses both this inference and his knowledge of $w_{G}$ to form a
conditional belief over $\theta _{2}$, and hence $q$.\footnote{%
Note that $\phi $ and $w_{G}$ do not form a clique in $G$. Consequently, $%
p_{G}(q\mid \phi ,w_{G})$ need not satisfy the unbiased-on-average property,
even if $G$ is perfect (see Spiegler (2020b)). Since we do not use this
property in the sequel, we also drop the assumption that $\mathcal{G}$
consists of perfect DAGs.}

The basic result that interior equilibrium fully reveals $\theta $ continues
to hold in this extended model. That is, in an interior equilibrium $h$, $%
\theta ^{\prime }\neq \theta $ implies $h(\theta ^{\prime },w^{\prime })\neq
h(\theta ,w)$. The proof is the same as in the case of Proposition \ref%
{fulll reveal}, and therefore omitted. However, the next result establishes
that equilibrium prices can \textit{also} reflect consumers' private
information, even when we hold $\theta $ fixed.

The result relies on the following notion of path blocking (in the spirit of
similar definitions in the literature on graphical probabilistic models ---
see Pearl (2009)). We say that a set of nodes $M$ \textit{blocks all
non-directed paths} between nodes $i,j\notin M$ if in the non-directed
version of $G$ (in which we ignore the direction of links), every path
between $i$ and $j$ passes through some $k\in M$. For example, in the DAG (%
\ref{extended short chain}), $\{\theta _{2}\}$ blocks all non-directed paths
between $w_{G}$ and $q$, whereas $\{\theta _{1}\}$ does not.\bigskip

\begin{proposition}
\label{prop fundamental}Suppose $\mathcal{G}$ is a set of DAGs that includes
non-rational consumer types. Moreover, suppose that for every non-rational $%
G\in \mathcal{G}$, $R(\phi )$ does not block all non-directed paths between $%
w_{G}$ and $q$. Then, assuming the interior equilibrium $h$ does not
coincide with REE, there must be a state $\theta $ and signals $w,w^{\prime }
$, such that $h(\theta ,w)\neq h(\theta ,w^{\prime })$.\bigskip 
\end{proposition}

Thus, the presence of imperfectly discerning consumers can create excessive
price fluctuations, in the sense that equilibrium prices respond to factors
beyond economic fundamentals. Specifically, they can reflect consumers'
private information, whereas this would not happen if consumers had rational
expectations. For instance, When $G$ is given by (\ref{extended short chain}%
), equilibrium prices respond to $w_{G}$ because consumers do not infer $%
\theta _{2}$ from prices.

\section{Related Literature}

Our paper contributes to a small literature on competitive markets with
asymmetric information, in which consumers' beliefs systematically deviate
from rational expectations. The closest precedent is Eyster and Piccione
(2013), who study dynamic competitive markets for financial securities
without short-selling, in which traders have diversely coarse models of the
market environment. This class of subjective models is similar to what we
assumed in Section 2, and a subset of our DAG-based formalism. Apart from
the different economic settings --- a dynamic financial market vs. a static
consumer market --- the main difference between Eyster and Piccione (2013)
and the present model is that traders in the Eyster-Piccione model do not
draw any inferences from current prices, whereas the heart of our model is
consumers' imperfect attempt to infer latent variables from current prices.

Piccione and Rubinstein (2003) analyze a simple example of a dynamic,
complete-information competitive market, in which producers differ in their
ability to perceive temporal price patterns, and hence in their ability to
predict market prices when making costly production decisions. They
demonstrate \textquotedblleft the existence of equilibrium fluctuations that
are unrelated to fundamentals...\textquotedblright\ (Piccione and Rubinstein
(2003, p. 218)), thus offering a precursor to Section 5.3 in our paper.

Our model fits naturally into the Behavioral Industrial Organization
literature (see Spiegler (2011) for a textbook treatment and Heidhues and K%
\H{o}szegi (2018) for a review). A prominent strand in this literature
analyzes market competition when firms use hidden fees (or other latent
product features) as part of their competitive strategy. Most of this
literature (going back to Gabaix and Laibson (2006)) has assumed that
consumers are unaware of the hidden charges and evaluate market alternatives
as if they do not exist.\footnote{%
In spiegler (2006), products have many dimensions, and consumers base their
product evaluation on a single, randomly drawn dimension.}

More generally, the behavioral IO literature has mostly assumed that
consumers have rational expectations, or that they have no understanding at
all of firms' incentives and therefore make no inferences from observations
that in fact indicate firms' attempts to exploit consumers' biases or
limitations. A few exceptions have examined market models in which consumers
have a \textit{coarse} understanding of market prices. Spiegler (2011, Ch.
8) synthesizes examples of bilateral-trade models with adverse selection
(extracted from Eyster and Rabin (2005), Jehiel and Koessler (2008), and
Esponda (2008)), in which the uninformed party has a coarse perception of
price formation.\footnote{%
In Esponda (2008), as in the present paper, consumers' assessment of firms'
types is based on the empirical distribution of \textit{active }firms at the
equilibrium price. There is no aggregate uncertainty in Esponda's model and
therefore no need to ask how consumers infer an aggregate state from
equilibrium prices.} At the extreme, this agent's belief is entirely coarse,
such that he correctly perceives average prices without having any
understanding of how they depend on the state of Nature.

In a similar vein, Murooka and Yamashita (2023) study a bilateral-trade
setting in which, with some probability, the buyer believes that product
quality is independent of the price in which it is traded. Ispano and
Schwardmann (2023) study a model in which consumers fail to understand that
only high-quality firms have an incentive to disclose their quality.
Schumacher (2023) studies a model in which firms sell a superior product
that only charges a base price and an inferior product that also includes an
add-on component. Coarse consumers know the average add-on charge across
products but incorrectly believe it is independent of the product type.
Thus, as in our model, consumers are aware of hidden charges but have
limited ability to predict them based on their information. Antler (2023)
analyzes a model of multilevel marketing and pyramid schemes, where a
principal exploits a network of agents having coarse expectations regarding
the network formation process, in the spirit of Jehiel (2005). These models
are all game-theoretic, and they lack the crucial feature of the present
paper, namely consumers' heterogeneous ability to draw inferences from
equilibrium prices.

\section{Conclusion}

The standard theory of competitive markets gives a central role to
equilibrium prices' ability to aggregate information. This property,
however, relies on market participants' ability to decipher the price
signal. This paper developed a new model of a competitive market in which
consumers differ in this regard, and explored the theoretical implications
of this \textquotedblleft cognitive friction\textquotedblright\ for the way
equilibrium outcomes respond to exogenous shocks.

The paper's methodological contribution inheres in our novel supply function
(arising from firms' differential ability to realize state-dependent latent
profit), our model of how consumers infer latent quantities from equilibrium
prices, and the tractable \textquotedblleft Bellman\textquotedblright\
characterization of interior equilibrium. The paper's substantive
conclusions include the deviation of equilibrium prices and add-ons from
their rational-expectations benchmarks, the equilibrium shift from latent to
salient price components as the set of consumer types expands, and the
demonstration that market outcomes respond to exogenous variables in ways
that are impossible under rational expectations.

We conclude the paper with a discussion of some of our modeling
procedures.\bigskip 

\noindent \textit{The homogenous-preference limit}

\noindent Our analysis in this paper has focused on the $\varepsilon
\rightarrow 0$ limit. A criticism of this approach is that on one hand our
full-revelation result (Proposition \ref{fulll reveal}) relies on preference
heterogeneity, yet our equilibrium analysis studies what happens when this
heterogeneity is almost non-existent. A counter-argument is that our
procedure is analogous to a common practice in the repeated games literature
(e.g., Mailath and Samuelson (2006)): Assuming players apply a discount
factor $\delta $ to future payoffs and then studying equilibria in the $%
\delta \rightarrow 1$ limit. The justification for that procedure is that
while discounting captures a key behavioral motive in long-term
interactions, assuming that this motive is weak enables a simple, clean
understanding of the logic of long-run cooperation. Likewise in our context,
preference heterogeneity allows equilibrium prices to reflect supply-side
responses to external shocks. This ensures that consumers' task of
deciphering equilibrium prices is meaningful. At the same time, assuming
weak taste heterogeneity enables us to focus on consumers' diverse add-on
forecasts.\bigskip

\noindent \textit{Non-uniformly distributed }$\pi $

\noindent The assumption that firm types $\pi $ are uniformly distributed
plays a facilitating role in our analysis, because it generates a linear
supply function. The Bellman-like equation (\ref{bellman}) arises from the
combination of two equations: The indifference condition for the \textit{%
marginal} firm type $\pi ^{\ast }(\theta ,h(\theta ))$, and consumers'
maximal willingness to pay for the product in state $\theta $ (which is
equal to $h(\theta )$ in the $\varepsilon \rightarrow 0$ limit). The latter
equation involves the \textit{average} active firm type $\bar{\pi}(\theta
^{\prime },h(\theta ^{\prime }))$ in various states $\theta ^{\prime }$.
When $\pi $ is uniformly distributed, $\pi ^{\ast }(\theta ^{\prime
},h(\theta ^{\prime }))$ and $\bar{\pi}(\theta ^{\prime },h(\theta ^{\prime
}))$ are linearly related, which enables us to conveniently substitute one
for the other. This also ensures that (\ref{bellman}) defines a contraction
mapping. If $\pi $ does not obey a uniform distribution, the tractable
linear structure of (\ref{bellman}) is lost, and a generalization of the
Bellman-like form will replace it. However, as long as the deviation from a
uniform distribution is not too large, the equilibrium equations will
continue to define a contraction mapping, such that the uniqueness of
interior equilibrium will prevail.

\noindent {\LARGE Appendix: Omitted Proofs}\bigskip 

\noindent {\large Proposition \ref{characterization}}

\noindent Equation (\ref{bellman}) is an immediate consequence of (\ref%
{eq_max_WTP}) and (\ref{price and addon}). By definition, $\overline{q}%
(\theta )\in \lbrack \frac{1}{2}S(\theta ),S(\theta )]$ for every $\theta $.
Thanks to the $\frac{1}{2}$ coefficient on the R.H.S of (\ref{bellman}), it
is then clear that the equation defines a contraction mapping over a compact
and convex Euclidean space. By the contraction mapping theorem, it has a
unique solution. This also uniquely pins down the values of $h(\theta )$ and 
$\pi ^{\star }(\theta ,h(\theta ))$ for every $\theta $.

We now obtain the bounds on $\bar{q}(\theta )$. Equation (\ref{bellman})
implies $2\bar{q}(\theta )\leq \max_{\theta }S(\theta )-\Delta +\max_{\theta
^{\prime }}\bar{q}(\theta ^{\prime })$ for every $\theta $. Therefore, $%
2\max_{\theta }\bar{q}(\theta )\leq \max_{\theta }S(\theta )-\Delta
+\max_{\theta }\bar{q}(\theta )$, such that $\max_{\theta }\bar{q}(\theta
)\leq S^{\max }-\Delta $. Likewise, (\ref{bellman}) implies $2\bar{q}(\theta
)\geq \min_{\theta }S(\theta )-\Delta +\min_{\theta ^{\prime }}\bar{q}%
(\theta ^{\prime })$ for every $\theta $. Therefore, $2\min_{\theta }\bar{q}%
(\theta )\geq \min_{\theta }S(\theta )-\Delta +\min_{\theta }\bar{q}(\theta
) $, such that $\min_{\theta }\bar{q}(\theta )\geq S^{\min }-\Delta $.

It remains to show that $\pi ^{\star }(\theta ,h(\theta ))\in (0,1)$ for
every $\theta $ --- i.e., the equilibrium is interior. Equivalently, we need
to show that for every $\theta $, $\frac{1}{2}S(\theta )<\bar{q}(\theta
)<S(\theta )$. Assume $\bar{q}(\theta )\geq S(\theta )$ for some $\theta $.
Then, (\ref{bellman}) implies%
\[
S(\theta )-\Delta +\min_{M\in \mathcal{M}}\sum_{\theta ^{\prime }}\mu
(\theta ^{\prime }\mid \theta _{M})\bar{q}(\theta ^{\prime })\geq 2S(\theta
) 
\]%
Since $\bar{q}(\theta ^{\prime })\leq S^{\max }-\Delta $ for every $\theta
^{\prime }$, $S^{\max }-S(\theta )\geq 2\Delta $. By definition, this means $%
S^{\max }-S^{\min }\geq 2\Delta $, contradicting (\ref{condition theorem}).
Therefore, $\bar{q}(\theta )<S(\theta )$ for every $\theta $. Now assume $%
\bar{q}(\theta )\leq \frac{1}{2}S(\theta )$ for some $\theta $. Then, (\ref%
{bellman}) implies%
\[
S(\theta )-\Delta +\min_{M\in \mathcal{M}}\sum_{\theta ^{\prime }}\mu
(\theta ^{\prime }\mid \theta _{M})\bar{q}(\theta ^{\prime })\leq S(\theta ) 
\]%
Since $\bar{q}(\theta ^{\prime })\geq S^{\min }-\Delta $ for every $\theta
^{\prime }$, $S^{\min }-2\Delta \leq 0$, contradicting (\ref{condition
theorem}). $\blacksquare $\bigskip

\noindent {\large Proposition \ref{prop average S}}

\noindent Suppose $\mathcal{M}=\{M\}$, where $M$ is arbitrary. Take an
expectation of both sides of (\ref{bellman}) with respect to $\mu $. Then,%
\[
2\sum_{\theta }\mu (\theta )\bar{q}(\theta )=\sum_{\theta }\mu (\theta
)S(\theta )-\Delta +\sum_{\theta ^{\prime }}\sum_{\theta }\mu (\theta )\mu
(\theta ^{\prime }\mid \theta _{M}^{\prime }=\theta _{M})\bar{q}(\theta
^{\prime }) 
\]%
As we observed above,%
\[
\sum_{\theta }\mu (\theta )\mu (\theta ^{\prime }\mid \theta _{M}^{\prime
}=\theta _{M})=\mu (\theta ^{\prime }) 
\]%
Therefore, the expected Bellman equation becomes%
\[
2\sum_{\theta }\mu (\theta )\bar{q}(\theta )=\sum_{\theta }\mu (\theta
)S(\theta )-\Delta +\sum_{\theta }\mu (\theta )\bar{q}(\theta ) 
\]%
such that $\sum_{\theta }\mu (\theta )\bar{q}(\theta )=\bar{S}-\Delta $,
which is the REE level. Thus, for any singleton $\mathcal{M}$, the expected
add-on in interior equilibrium coincides with its REE level. By Proposition (%
\ref{prop adding types}), for any $\mathcal{M}\supset \{M\}$, the expected
add-on level in interior equilibrium is weakly lower in each state than
under $\{M\}$. It follows that the ex-ante expected add-on is weakly below
the REE level $\bar{S}-\Delta $.

Recall that by (\ref{quality}), $\bar{q}(\theta )\geq \frac{1}{2}S(\theta )$
for every $\theta $. Therefore, the ex-ante expected add-on cannot fall
below $\frac{1}{2}\bar{S}$. We now construct primitives $\Theta ,\mu ,S,%
\mathcal{M}$ that satisfy $\sum_{\theta }\mu (\theta )S(\theta )=\bar{S}$
and condition (\ref{condition theorem}), and show that the expected add-on
in the interior equilibrium under this specification is arbitrarily close to 
$\frac{1}{2}\bar{S}$. Let $\theta _{i}\in \{0,1\}$ for every $i=1,...,n$,
where $n$ is arbitrarily large. Let $e_{i}$ denote the state $\theta $ for
which $\theta _{i}=0$ and $\theta _{j}=1$ for all $j\neq i$. distribution $%
\mu $ as follows: $\mu (0,...,0)=\alpha $ and $\mu (e_{i})=(1-\alpha )/n$
for every $i$. We will pin down $\alpha $ below. Define the function $S$ as
follows: $S(0,...,0)\gtrapprox 2\Delta $, and $S(e_{i})\lessapprox 4\Delta $
for every $i=1,...,n$. Fix $\alpha $ such that $\bar{S}\approx \alpha \cdot
2\Delta +(1-\alpha )\cdot 4\Delta $, i.e., $\alpha \approx (4\Delta -\bar{S}%
)/2\Delta $. Finally, let $\mathcal{M}$ consist of the following types: the
rational type $\{1,...,n\}$, and the coarse types $\{i\}$ for every $%
i=1,...,n$.

Guess an equilibrium in which the rational type buys the product in the
state $(0,...,0)$; and the coarse type $\{i\}$ buys the product in the state 
$e_{i}$, for every $i=1,...,n$. The Bellman-like equations are thus reduced
to%
\[
2\bar{q}(0,...,0)=S(0,...,0)-\Delta +\bar{q}(0,...,0) 
\]%
and%
\[
2\bar{q}(e_{i})=S(e_{i})-\Delta +\frac{\alpha }{\alpha +\frac{1-\alpha }{n}}%
\bar{q}(0,...,0)+\frac{\frac{1-\alpha }{n}}{\alpha +\frac{1-\alpha }{n}}\bar{%
q}(e_{i}) 
\]%
for every $i=1,...,n$. It follows that $\bar{q}(0,...,0)=S(0,...,0)-\Delta
\approx \Delta $; and as $n\rightarrow \infty $, the solution to the
remaining equations is $\bar{q}(e_{i})\approx 2\Delta $. It is
straightforward to confirm that the types that buy the product in each state
have the lowest add-on estimate in that state. The ex-ante equilibrium
add-on is approximately%
\[
\alpha \cdot \Delta +(1-\alpha )\cdot 2\Delta \approx \frac{1}{2}\bar{S} 
\]%
as required. $\blacksquare $\bigskip

\noindent {\large Proposition \ref{proposition range}}

\noindent A rational type's willingness to pay in state $\theta $ is $%
v^{\ast }-\bar{q}(\theta )$. Therefore, $h(\theta )\geq v^{\ast }-\bar{q}%
(\theta )$ for every $\theta $. Plugging the upper bound on $\bar{q}(\theta
) $ given by Proposition \ref{characterization}, we obtain%
\[
h(\theta )\geq v^{\ast }-(S^{\max }-\Delta )=2v^{\ast }-c-S^{\max } 
\]

Now consider the state $\theta $ for which $S(\theta )=S^{\min }$. The
rational type's willingness to pay in this state is $v^{\ast }-\bar{q}%
(\theta )$. The willingness to pay of an arbitrary type $M$ is 
\begin{equation}
v^{\ast }-\sum_{\theta ^{\prime }}\mu (\theta ^{\prime }\mid \theta _{M})%
\bar{q}(\theta ^{\prime })  \label{type M WTP}
\end{equation}

Guess a solution to (\ref{bellman}) for which $\bar{q}(\theta )=S^{\min
}-\Delta $. Then, the rational type's willingness to pay in state $\theta $
is $v^{\ast }-(S^{\min }-\Delta )$. By the lower bound on $\bar{q}(\theta )$
given by Proposition \ref{characterization}, this expression is weakly above
(\ref{type M WTP}) for any $M$. Then, guessing that the rational type has
the highest willingness to pay in $\theta $ is consistent with a solution to
(\ref{bellman}) in this state, and it gives $h(\theta )=v^{\ast }-(S^{\min
}-\Delta )$. The remaining equations in (\ref{bellman}) for all other states
deliver a unique solution, hence the guess is consistent with the entire
system of equations. It follows that when $\mathcal{M}$ contains a rational
type, the upper bound on equilibrium prices given in part $(ii)$ is binding. 
$\blacksquare $\bigskip 

\noindent {\large Proposition \ref{mutually_beneficial_welfare}}

\noindent In REE, consumer welfare is null as 
\[
h(\theta )=v^{\star }+\bar{q}(\theta )=v^{\star }+\frac{1+\pi ^{\star
}(\theta ,h(\theta ))}{2}S(\theta ) 
\]%
Let $\pi ^{\prime }=(1+\pi ^{\star }(\theta ,h(\theta ))S(\theta )/2$. Trade
with any type $\pi \in \lbrack \pi ^{\star }(\theta ,h(\theta )),\pi
^{\prime })$ yields a welfare loss to the consumer. Since, by definition,
firms of type $\pi ^{\star }(\theta ,h(\theta ))$ earn zero in equilibrium
in state $\theta $, continuity implies that there is a cutoff $\pi ^{\prime
\prime }\in (\pi ^{\star }(\theta ,h(\theta )),\pi ^{\prime })$ such that
trade with firms of type $\pi <\pi ^{\prime \prime }$ is socially harmful.
By Proposition \ref{mutually_beneficial_addition}, expanding the set of
cognitive types weakly increases the equilibrium price in each state. Hence,
the cutoff $\pi ^{\star }(\theta ,h(\theta ))$ weakly decreases in every
state. It follows that the measure of firms that trade in equilibrium weakly
goes up in every state. $\blacksquare $\bigskip

\noindent {\large Lemma \ref{lemma beta}}

\noindent Since $G$ is perfect, there is an equivalent DAG $G^{\prime }$ (in
the sense that $p_{G}\equiv p_{G^{\prime }})$ in which $\phi $ is an
ancestral node (see Spiegler (2020a,b)). Therefore, we can regard $\phi $ as
ancestral, without loss of generality. If there is a direct link $\phi
\rightarrow q$, then $(\phi ,q)$ form a clique in $G$, and hence perfection
implies $p_{G}(q,\phi )\equiv p(q,\phi )$, hence $p_{G}(q\mid \phi )\equiv
p(q\mid \phi )$ whenever $p(\phi )>0$. Since $p$ is fully revealing, $%
p(q\mid \phi )\equiv p(q\mid \theta _{p}(\phi ))$. In this case, the unique
transition matrix $\beta $ for which (\ref{beta representation}) holds is $%
\beta (\theta \mid \theta )\equiv 1$.

Now suppose there is no path from $q$ to $\phi $. Then, $q$ is independent
of $\phi $ according to $p_{G}$, such that $p_{G}(q\mid \phi )\equiv p(q)$.
We can thus rewrite%
\[
p_{G}(q\mid \phi )=\sum_{\theta ^{\prime }}p(\theta ^{\prime })p(q\mid
\theta ^{\prime }) 
\]%
In this case, the unique transition matrix $\beta $ for which (\ref{beta
representation}) holds is $\beta (\theta ^{\prime }\mid \theta )\equiv \mu
(\theta ^{\prime })$.

Now suppose there is a path from $q$ to $\phi $, but the two nodes are not
directly related. Note that all nodes along all paths from $q$ to $\phi $
represent $\theta $ variables. Let $C$ denote the set of nodes to which $%
\phi $ sends direct links, and let $D$ denote the set of nodes that send
direct links into $q$. Then,%
\[
p_{G}(q\mid \phi )=\sum_{\theta _{C}}p(\theta _{C}\mid \phi )\sum_{\theta
_{D}}p_{G}(\theta _{D}\mid \theta _{C})p(q\mid \theta _{D}) 
\]%
Since $p$ is fully revealing, $p(\theta _{M}\mid \phi )$ assigns probability
one to the projection of $\theta _{p}(\phi )$ on the variables represented
by $C$, denoted $\theta _{C}(\phi )$.

Therefore, 
\[
p_{G}(q\mid \phi )=\sum_{\theta _{D}}p_{G}(\theta _{D}\mid \theta _{C}(\phi
))p(q\mid \theta _{D})=\sum_{\theta _{D}}p_{G}(\theta _{D}\mid \theta
_{C}(\phi ))\sum_{\theta ^{\prime }}p(\theta ^{\prime }\mid \theta
_{D})p(q\mid \theta ^{\prime })
\]%
Denote%
\[
\beta (\theta ^{\prime }\mid \theta )=\sum_{\theta _{D}^{\prime \prime
}}p_{G}(\theta _{D}^{\prime \prime }\mid \theta _{C})\sum_{\theta ^{\prime
}}p(\theta ^{\prime }\mid \theta _{D})
\]%
The R.H.S of this equation is pinned down by $G$ and $\mu $. Thus, it is the
unique transition matrix for which (\ref{beta representation}) holds.
Moreover, the property that $\mu $ is an invariant distribution of $\beta $
is an immediate consequence of (\ref{no distortion of average quality}). $%
\blacksquare $\bigskip 

\noindent {\large Proposition \ref{prop fundamental}}

\noindent\ Assume the contrary --- i.e., $\mathcal{G}$ satisfies the
premises of the result, and yet the interior equilibrium $h$ is purely a
function of $\theta $. By assumption, $h(\theta )$ deviates from the REE
price in some $\theta $. In that state,%
\[
h(\theta )=v^{\ast }-\int_{q}p_{G}(q\mid h(\theta ),w_{G})q
\]%
for all realizations of $w_{G}$ --- since by assumption, $h$ is unresponsive
to $w_{G}$ given $\theta $. However, by assumption, the DAG $G$ has the
property that there is a non-directed path between $w_{G}$ and $q$ that does
not pass through any node in $R(\phi )$. For generic distributions $\mu $,
this means that type $G$'s belief over $q$ conditional on $h(\theta )$ is
not invariant to $w_{G}$, hence this type's willingness to pay varies with $%
w_{G}$ given $\theta $. As a result, the equilibrium price cannot be
constant in $w_{G}$ given $\theta $, a contradiction. $\blacksquare $

\end{document}